\documentclass [11pt, a4paper] {article}
\usepackage{fullpage}
\newcommand {\OelOqGyI} {On quantum interactive proofs with short messages}
\usepackage{amssymb}
\usepackage{amsfonts}
\usepackage{bbm}
\usepackage{hyperref}
\usepackage{amsthm}
\usepackage{mathtools}
\usepackage{algorithmic}
\usepackage{algorithm}
\usepackage{asymptote}

\newcommand {\SddsvkvC} [1] {\texttt{#1}}
\newcommand {\chmyOupr} {Attila Pereszl\'{e}nyi}
\newcommand {\KfiPvLSr} {\SddsvkvC{attila.pereszlenyi@gmail.com}}
\newcommand {\kFJIkEiy}
{Centre for Quantum Technologies, National University of Singapore}
\newcommand {\UdxPemrT} {Without loss of generality}
\newcommand {\EtVQAUvQ} {without loss of generality}
\newcommand {\bsIbxBHW} {if and only if}
\newcommand {\keQZCenD} {such that}
\newcommand {\UgLlztAi} {I.e.,}
\newcommand {\DwcWexfp} {i.e.,}
\newcommand {\WvEhIjGd} {Hilbert space}
\newcommand {\LsetZeCP} {Choi--Jamio{\l}kowski representation}
\newcommand {\eyBJtrpd} {et al.}
\newcommand {\KBvQOvrw} {Quantum interactive proof system}
\newcommand {\JgSMGwES} {quantum interactive proof system}
\newcommand {\UeDBErYM} {
\author {\chmyOupr\thanks{E-mail: \KfiPvLSr.}\\
\textsl{\small \kFJIkEiy}}
}
\newcommand {\MGEHGnhN} {
\bibliographystyle{halpha}
\bibliography{./bib}
}
\newcommand {\HAPnvheP} [1] {\ensuremath{ \left( #1 \right) }}
\newcommand {\OLNiNHyK} [1] {\ensuremath{ \left[ #1 \right] }}
\newcommand {\ZcdnmWHx} [1] {\ensuremath{ \left\lbrace #1 \right\rbrace }}
\newcommand {\MwLfFFxk} {\: \! \!}
\newcommand {\GPnJfPwC} [2] {\ensuremath{ #1 \MwLfFFxk \HAPnvheP{ #2 } }}

\newcommand {\XYYXlqUI} {\ensuremath{ \stackrel {\mathrm{def}} {=} }}
\newcommand {\UpcxRgfA} [1] {\GPnJfPwC {O} {#1}}
\newcommand {\WctPVadE} [1] {\UpcxRgfA{\log #1}}
\newcommand {\BXycapFx} [1] {\ZcdnmWHx{ 1, 2, \ldots, #1 }}
\newcommand {\orxHiCqI} [1] {\ensuremath{ \left\lceil #1 \right\rceil }}
\newcommand {\ZVVBqftq} [3] {\ensuremath{ #1 : \, #2 \rightarrow #3 }}
\newcommand {\vpmfnZuX} {\ensuremath{ \mathsf{poly} }}
\newcommand {\nkZOryEx} [1] {\GPnJfPwC {\vpmfnZuX} {#1}}
\newcommand {\EtgXbDyQ} {\ensuremath{ \mathsf{log} }}
\newcommand {\OraXXrdX} [1] {\GPnJfPwC {\EtgXbDyQ} {#1}}
\newcommand {\aYnrcuNB} [1] {\ensuremath{ \mathbb{#1} }}
\newcommand {\oKlAveag} {\aYnrcuNB{R}}
\newcommand {\sFPODjqO} {\aYnrcuNB{C}}
\newcommand {\bSkTgkBm} {\aYnrcuNB{N}}
\newcommand {\pRWLPcKq} [1] {\ensuremath{ \left\langle #1 \right| }}
\newcommand {\CligxKFo} [1] {\ensuremath{ \left| #1 \right\rangle }}
\newcommand {\clWmGeKU} [2] {\ensuremath{ \CligxKFo{#1} \! \pRWLPcKq{#2} }}
\newcommand {\sCUPypWN} [1] {\clWmGeKU{#1}{#1}}

\newcommand {\RoHOARwu} {\ensuremath{ \otimes }}
\newcommand {\znbdCPOj} [1] {\ensuremath{ #1^{*} }}
\newcommand {\VGXhOmeS} [1] {\ensuremath{ #1^{-1} }}
\newcommand {\FyjDddRU} {\ensuremath{ \mathrm{Tr} }}
\newcommand {\jYPzSEcs} [2] {\GPnJfPwC{ \FyjDddRU_{#1} }{ #2 }}
\newcommand {\xxsmCMTw} [1] {\ensuremath{ \left\| #1 \right\| }}
\newcommand {\BVOiHqKx} [2] {\ensuremath{ \xxsmCMTw{#1}_{#2} }}
\newcommand {\RtRjhAaz} [1] {\BVOiHqKx{#1}{\FyjDddRU}}
\newcommand {\fsXxynbX} [2] {\ensuremath{ \frac{1}{2} \RtRjhAaz{#1 - #2} }}
\newcommand {\lScNDDFa} [1] {\ensuremath{ \mathsf{#1} }}
\newcommand {\uplxoTuD} {\lScNDDFa{NP}}
\newcommand {\MVZVnART} {\lScNDDFa{PSPACE}}
\newcommand {\JKkcVYtS} {\lScNDDFa{IP}}
\newcommand {\OzghOdjr} {\lScNDDFa{MA}}
\newcommand {\VFQyzTUk} {\lScNDDFa{QMA}}
\newcommand {\WfUTAOJK} {\lScNDDFa{QIP}}
\newcommand {\ObjIVIHZ} {\lScNDDFa{BQP}}

\newcommand {\HiQndmIP} [1] {\ensuremath{ \mathnormal{#1} }}
\newcommand {\FdVGMifg} [1] {\ensuremath{ \mathbf{#1} }}
\newcommand {\WUmziIIO} {\FdVGMifg{CNOT}}
\newcommand {\SRoWkgaL} {\FdVGMifg{H}}
\newcommand {\hWDmTLuj} {\FdVGMifg{T}}
\newcommand {\fAEoNWeb} {\ensuremath{ \mathbbm{1} }}
\newcommand {\nGRKiFQI} [1] {\ensuremath{ \mathsf{#1} }}
\newcommand {\zkKUsJUE} [1] {\ensuremath{ \mathcal{#1} }}
\newcommand {\CrzaBbkD} [1] {\ensuremath{\mathrm{#1}}}
\newcommand {\xMawpTFS} [1] {\GPnJfPwC{\CrzaBbkD{L}}{#1}}
\newcommand {\ppKvmYcj} [1] {\GPnJfPwC{\CrzaBbkD{D}}{#1}}
\newcommand {\zRYfDEHl} [2] {\GPnJfPwC{\CrzaBbkD{C}}{#1, #2}}
\theoremstyle {plain}
\newtheorem {thm} {Theorem} [section]
\newtheorem {cor} [thm] {Corollary}
\newtheorem {lem} [thm] {Lemma}
\theoremstyle {remark}
\newtheorem {rem} [thm] {Remark}

\theoremstyle {definition}
\newtheorem {defi} [thm] {Definition}

\hypersetup{
pdftitle={\OelOqGyI},
pdfauthor={\chmyOupr}
}
\newcommand {\QMGFPyuL} [3] {\GPnJfPwC{ \WfUTAOJK_{\mathrm{short}} }{ #1, #2, #3 }}
\newcommand {\sMiihMKK} [2] {\GPnJfPwC{ \WfUTAOJK }{ \OLNiNHyK{\EtgXbDyQ, \vpmfnZuX}, #1, #2 }}
\begin{asydef}
void curly_bracket(real x, real ybottom, real ytop){
real ym = (ybottom + ytop) / 2;
path unitbracket = (x, 0){right}..{right}(x + 1, 2);
path halfbracket = shift(0, ybottom) *
yscale((ym-ybottom) / 2) * unitbracket;
draw(halfbracket);
draw(reflect((0,ym), (1,ym))*halfbracket);
}
\end{asydef}
\title {\textbf{\OelOqGyI}}
\UeDBErYM
\date{September 5, 2011}
\begin{document}
\maketitle
\begin{abstract}
This paper proves one of the open problem
posed by Beigi \eyBJtrpd~in \cite{Beigi2011}.
We consider \JgSMGwES{}s where in the beginning the verifier
and prover send messages to each other with the combined length
of all messages being at most logarithmic (in the input length);
and at the end the prover sends a polynomial-length
message to the verifier.
We show that this class has the same expressive power as \VFQyzTUk.
\end{abstract}
\section{Introduction}
\KBvQOvrw{}s (\WfUTAOJK) were introduced by \cite{Watrous1999,Watrous2003}
as a natural extension of interactive proofs (\JKkcVYtS) to the
quantum computational setting.
They have been extensively studied and now it's known
that the power of \JgSMGwES{}s is the same as the classical one,
\DwcWexfp{} $ \WfUTAOJK = \JKkcVYtS = \MVZVnART $ \cite{Jain2010}.
Furthermore, \JgSMGwES{}s still has the same expressive power
if we restrict the number of messages to three and
have exponentially small one-sided error \cite{Kitaev2000}.
If the interaction is only one message from the prover to the
verifier then the class is called \VFQyzTUk, which is the
quantum analogue of \uplxoTuD{} and \OzghOdjr{}.
\VFQyzTUk{} can also be made to have exponentially small error,
and has natural complete problems \cite{Aharonov2002}.
\par
Several variants of \WfUTAOJK{} and \VFQyzTUk{} have also been studied.
We now focus on the case where some or all of the messages are
small, meaning at most logarithmic in the input length.
These cases are usually not interesting in the classical setting
since a logarithmic-length message can be eliminated by the
verifier by enumerating all possibilities.
This is not true in the quantum case,
indeed the variant of \VFQyzTUk{} where we have two unentangled
logarithmic-length proofs contains \uplxoTuD{} \cite{Blier2009};
hence not believed to be equal to \ObjIVIHZ.
On the other hand, if \VFQyzTUk{} has one logarithmic-length proof
then it has the same expressive power as \ObjIVIHZ{} \cite{Marriott2005}.
\par
Beigi \eyBJtrpd~\cite{Beigi2011} proved that in other variants of
\JgSMGwES{}s the short message can also be eliminated without
changing the power of the proof system.
Among others they showed that in the setting when the verifier
sends a short message to the prover and the prover responds
with an ordinary, polynomial-length message, the short message
can be discarded, and hence the class has the same power as \VFQyzTUk.
They have raised the question if this is also true if we replace
the short question of the verifier with a `short interaction'.
\UgLlztAi{} consider \JgSMGwES{}s where in the beginning the verifier
and prover send messages to each other with the combined length
of all messages being at most logarithmic, and at the end the
prover sends a polynomial-length message to the verifier.
We show that this class has the power of \VFQyzTUk, or in other words,
the short interaction can be discarded.
This is formalized by the following theorem.
\begin{thm}
Let $ c,s : \bSkTgkBm \rightarrow \HAPnvheP{0,1} $ be polynomial-time
computable functions \keQZCenD{}
$ \GPnJfPwC{c}{n} - \GPnJfPwC{s}{n} \in 1/\nkZOryEx{n} $.
Then $ \QMGFPyuL{\OraXXrdX{n}}{c}{s} = \VFQyzTUk $.
\label{thm:QIPshort_equals_QMA}
\end{thm}
Here \QMGFPyuL{\OraXXrdX{n}}{c}{s} is the class described
above, with completeness-soundness gap being separated by some
inverse polynomial function of the input length.
For a rigorous description of the class see
Definition~\ref{def:QIP_short}, and for the notation
see the discussion in Section~\ref{sec:preliminaries}.
\par
The remainder of the paper is organized as follows.
Section~\ref{sec:preliminaries} discusses the background
theorems and definitions needed for the rest of the paper.
The proof of the main theorem is split into two.
Section~\ref{sec:prover_space_compression} describes
how to deal with the short interaction, and
using that Section~\ref{sec:main_theorem_proof}
finishes the proof by discussing how to handle the
last round.
This part closely follows the corresponding
proof in \cite{Beigi2011}.
\section{Preliminaries}
\label{sec:preliminaries}
We assume familiarity with quantum information \cite{Watrous2008},
computation \cite{Nielsen2000} and computational complexity
\cite{Watrous2008a}; such as
mixed states, unitary operations, quantum channels,
representations of quantum channels, quantum de Finetti
theorems, state tomography and complexity classes like
\VFQyzTUk{} and \WfUTAOJK.
The purpose of this section is to present the notations
and background information (definitions,
theorems) required to understand the following
two sections.
\par
We denote the set of functions of $n$ that are
upper-bounded by some polynomial in $n$ by \nkZOryEx{n}.
If the argument is clear, we omit it and just
write \vpmfnZuX.
Similarly, we write \OraXXrdX{n} or \EtgXbDyQ{} for
the set of functions that are in \WctPVadE{n}.
We try to follow the notations used in
\cite{Watrous2008,Beigi2011}.
When we talk about a quantum register (\nGRKiFQI{R})
of size $k$, we mean the object made up of
$k$ qubits.
It has associated \WvEhIjGd{}
$ \zkKUsJUE{R} = \sFPODjqO^{2^k} $.
\xMawpTFS{\zkKUsJUE{R}} denotes the space of all linear
mappings from \zkKUsJUE{R} to itself,
and the set of all density operators on \zkKUsJUE{R}
is denoted by \ppKvmYcj{\zkKUsJUE{R}}.
The adjoint of $ \FdVGMifg{X} \in \xMawpTFS{\zkKUsJUE{R}} $
is denoted by \znbdCPOj{\FdVGMifg{X}},
and the trace norm of \FdVGMifg{X} by
\RtRjhAaz{\FdVGMifg{X}}.
The trace distance between \FdVGMifg{X} and \FdVGMifg{Y}
is defined as
\[ \fsXxynbX{\FdVGMifg{X}}{\FdVGMifg{Y}} . \]
\par
A quantum channel ($\Phi$) is a completely positive
and trace-preserving linear map of the form
\ZVVBqftq{\Phi}{\xMawpTFS{\zkKUsJUE{Q}}}{\xMawpTFS{\zkKUsJUE{R}}}.
The set of all such channels is denoted by
\zRYfDEHl{\zkKUsJUE{Q}}{\zkKUsJUE{R}}.
For any $ \Phi \in \zRYfDEHl{ \sFPODjqO^{2^k} }
{ \sFPODjqO^{2^{\ell}} } $ the normalized \LsetZeCP{}
of $ \Phi $ is defined to be
\[ \rho_{\Phi} \in
\ppKvmYcj{ \sFPODjqO^{2^{\ell}} \RoHOARwu \sFPODjqO^{2^k} },
\quad \quad \rho_{\Phi} =
\frac{1}{2^k} \sum_{ x,y \in \ZcdnmWHx{0,1}^k }
\GPnJfPwC{\Phi}{\clWmGeKU{x}{y}} \RoHOARwu \clWmGeKU{x}{y}. \]
It can be generated by applying $\Phi$ on
one half of $k$ pairs of qubits in the state
$ \frac{1}{\sqrt{2}} \HAPnvheP{\CligxKFo{00} + \CligxKFo{11}} $.
If we are given $\rho_{\Phi}$ and an arbitrary
$ \sigma \in \ppKvmYcj{\sFPODjqO^{2^k}} $ then
there exist a simple procedure which produces
\GPnJfPwC{\Phi}{\sigma} with probability
$ 1 / 4^k $.
We will refer to it as `post-selection'.
For details see \cite[Section 2.1]{Beigi2011}.
\par
When we talk about a polynomial-time
quantum algorithm, we mean a quantum circuit
containing Hadamard (\SRoWkgaL),
$ \pi / 8 $ (\hWDmTLuj) and controlled-not
(\WUmziIIO) gates, and which can be generated
by a classical algorithm in polynomial-time.
The classes \VFQyzTUk{} and \WfUTAOJK{} have been defined
in \cite{Aharonov2002} and \cite{Watrous2003}
respectively, and we will use those definitions.
Now we want to define the \JgSMGwES{}s where
in the beginning there is a \OraXXrdX{n}-long
interaction which is followed by a
\nkZOryEx{n}-length message from the prover.
Note that in this setting we can assume, \EtVQAUvQ{}, that
all the messages except the last one are one
qubits, and the total number of rounds is at most \WctPVadE{n}.
This is because we can add dummy qubits
that are interspersed with the qubits sent by the other party.
We define the class according to this observation.
\begin{defi}
Let the class \QMGFPyuL{m}{c}{s} be the set of languages for which there exist
a quantum interactive proof system with the following properties.
The completeness parameter is $c$ and the soundness is $s$.
The proof system consists of $m$ rounds, each round is a
question-answer pair.
All questions and answers are one qubits except for the last
answer which is \nkZOryEx{n} qubits,
where $n$ is the length of the input.
See Figure~\ref{fig:three_round_QIP} for an example
with $ m=3 $.
\label{def:QIP_short}
\end{defi}
\begin{figure}[h]
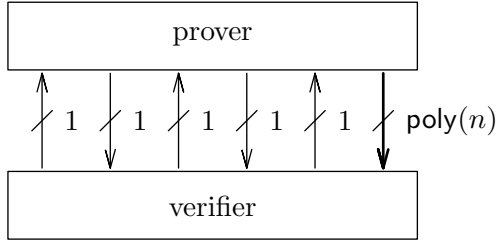

\centering
\begin{asy}
size(6.5cm);
defaultpen(fontsize(11));
void drawbox(int x, int y, string name){
draw((x, y)--(x+12, y)--(x+12, y+2)--(x, y+2)--cycle);
label(name, (x+6,y+1));
}
void drawarrow(int x, real y1, real y2, string name, pen p=currentpen, real size=7){
draw((x,y1)--(x,y2), p, EndArrow(SimpleHead, size));
pair m=(x, (y1+y2)/2);
real l=0.3;
draw((m.x-l, m.y-l)--(m.x+l, m.y+l));
label(name, (m.x + l + 0.1, m.y), E);
}
drawbox(0, 0, "verifier");
drawbox(0, 5, "prover");
drawarrow(1, 2.1, 4.9, "$1$");
drawarrow(3, 4.9, 2.1, "$1$");
drawarrow(5, 2.1, 4.9, "$1$");
drawarrow(7, 4.9, 2.1, "$1$");
drawarrow(9, 2.1, 4.9, "$1$");
pen p=currentpen + 1;
drawarrow(11, 4.9, 2.1, "$\mathsf{poly}(n)$", p);
\end{asy}
\caption{The interaction in the proof system of
Definition~\ref{def:QIP_short} in case $ m=3 $.}
\label{fig:three_round_QIP}
\end{figure}
A similar class, \sMiihMKK{c}{s} was defined in \cite{Beigi2011}
to be the class of problems for which there exist a one round
quantum interactive proof system, with completeness and
soundness parameters $c$ and $s$.
Additionally the verifier's question has length \OraXXrdX{n},
and the prover's answer is \nkZOryEx{n} qubits.
\begin{rem}
The following inclusion is trivially true between the above classes.
\[ \sMiihMKK{c}{s} \subseteq \QMGFPyuL{\OraXXrdX{n}}{c}{s} , \]
for all values of $c$ and $s$.
\end{rem}
In \cite{Beigi2011} it was proven that in their setting
the question from the verifier is unnecessary.
Or more precisely:
\begin{thm}[\cite{Beigi2011}]
Let $ c,s : \bSkTgkBm \rightarrow \HAPnvheP{0,1} $ be polynomial-time
computable functions \keQZCenD{}
$ \GPnJfPwC{c}{n} - \GPnJfPwC{s}{n} \in 1/\nkZOryEx{n} $.
Then $ \sMiihMKK{c}{s} = \VFQyzTUk $.
\label{thm:QIPlogpoly_equals_QMA}
\end{thm}
In the next sections we prove that the seemingly stronger
class of Definition~\ref{def:QIP_short} also has the power
of \VFQyzTUk{} if $ m = \WctPVadE{n} $.
For this we will need the following theorems.
\begin{thm}[\cite{Nielsen2000}, Chapter~4.5.2]
An arbitrary unitary operator on $\ell$ qubits can be implemented using a circuit
containing \UpcxRgfA{\ell^2 4^{\ell}} single qubit and \WUmziIIO{} gates.
\label{thm:unitary_to_circuit}
\end{thm}
The next theorem follows from the Solovay--Kitaev theorem
\cite{Kitaev1997}, and also appears in \cite{Nielsen2000}.
\begin{thm}
For any unitary operator \FdVGMifg{U} on one qubit
and $ \varepsilon > 0 $, there exist a circuit
$ \HiQndmIP{C}_{\FdVGMifg{U}, \varepsilon} $ \keQZCenD{}
$ \HiQndmIP{C}_{\FdVGMifg{U}, \varepsilon} $ is made up of
\UpcxRgfA{\GPnJfPwC{\log ^3}{1 / \varepsilon}} gates from the set
\ZcdnmWHx{\SRoWkgaL, \hWDmTLuj}, and for all $\CligxKFo{\varphi} \in \sFPODjqO^2$
it holds that
\[ \fsXxynbX{ \FdVGMifg{U} \sCUPypWN{\varphi} \znbdCPOj{\FdVGMifg{U}} }
{ \sCUPypWN{\xi} } \leq \varepsilon , \]
where $ \CligxKFo{\xi} \XYYXlqUI
\GPnJfPwC{\HiQndmIP{C}_{\FdVGMifg{U}, \varepsilon}}{\CligxKFo{\varphi}} $.
\label{thm:Solovay-Kitaev}
\end{thm}
The following is corollary to
Theorem~\ref{thm:unitary_to_circuit}
and \ref{thm:Solovay-Kitaev}.
\begin{cor}
For any unitary operator \FdVGMifg{U} on $\ell$ qubits
and $ \varepsilon > 0 $, there exist a circuit
$ \HiQndmIP{C}_{\FdVGMifg{U}, \varepsilon} $ \keQZCenD{}
$ \HiQndmIP{C}_{\FdVGMifg{U}, \varepsilon} $ is made up of
\UpcxRgfA{ 5^{\ell} \cdot \GPnJfPwC{\log ^3}{5^{\ell} / \varepsilon} }
gates from the set
\ZcdnmWHx{\SRoWkgaL, \hWDmTLuj, \WUmziIIO}, and for all $\CligxKFo{\varphi} \in \sFPODjqO^{2^{\ell}}$
it holds that
\[ \fsXxynbX{ \FdVGMifg{U} \sCUPypWN{\varphi} \znbdCPOj{\FdVGMifg{U}} }
{ \sCUPypWN{\xi} } \leq \varepsilon , \]
where $ \CligxKFo{\xi} \XYYXlqUI
\GPnJfPwC{\HiQndmIP{C}_{\FdVGMifg{U}, \varepsilon}}{\CligxKFo{\varphi}} $.
\label{cor:unitary_approx_by_circuit}
\end{cor}
\begin{lem}[Lemma 1 of \cite{Beigi2011}]
Let $ \rho \in \ppKvmYcj{\sFPODjqO^{2^q}} $ be a
state on $ q = \UpcxRgfA{\log n} $ qubits.
For any $ \varepsilon \in 1 / \nkZOryEx{n} $,
choose $N$ \keQZCenD{} $ N \geq 2^{10q} / \varepsilon^3 $
and $ N \in \nkZOryEx{n} $.
If $ \rho^{\RoHOARwu N} $ is given to a
\nkZOryEx{n}-time quantum machine, then it
can perform quantum state tomography,
and get a classical description
$ \xi \in \xMawpTFS{\sFPODjqO^{2^q}} $ of
$ \rho $, which with probability at least
$ 1 - \varepsilon $ satisfies
\[ \RtRjhAaz{\rho - \xi} < \varepsilon . \]
\label{lem:state_tomography}
\end{lem}
\begin{thm}[quantum de Finetti theorem \cite{Christandl2007};
this form is from \cite{Watrous2008}]
Let $ \nGRKiFQI{X}_1, \ldots, \nGRKiFQI{X}_n $ be identical quantum registers,
each having associated space $\sFPODjqO^d$,
and let $ \rho \in \ppKvmYcj{ \sFPODjqO^{dn} } $ be the state of these registers.
Suppose that for all permutation $ \pi \in S_n $ it holds that
$ \rho = \FdVGMifg{W}_{\pi} \rho \znbdCPOj{\FdVGMifg{W}_{\pi}} $,
where $ \FdVGMifg{W}_{\pi} $ permutes the contents of
$ \nGRKiFQI{X}_1, \ldots, \nGRKiFQI{X}_n $
according to $\pi$.
Then for any choice of $ k \in \ZcdnmWHx{2,3, \ldots, n-1} $
there exists a number $ N \in \bSkTgkBm $,
a probability vector $ p \in \oKlAveag^N $,
and a collection of density operators
$ \ZcdnmWHx{\sigma_i : i \in \BXycapFx{N}} \subset \ppKvmYcj{\sFPODjqO^d} $
such that
\[ \RtRjhAaz{ \rho^{\nGRKiFQI{X}_1 \cdots \nGRKiFQI{X}_k } -
\sum_{i=1}^{N} p_i \sigma_{i}^{\RoHOARwu k} }
< \frac{4 d^2 k}{n} . \]
\label{thm:deFinetti}
\end{thm}
\section{Compressing the prover's private space}
\label{sec:prover_space_compression}
This section proves that during the initial
short interaction the prover doesn't need to keep
too many qubits in its private memory.
This will make it easy to give the description of the
actions of the prover during these rounds as a
classical proof.
This idea of upper-bounding the prover's
private space has also appeared in \cite{Kobayashi2003a}.
\par
From now on, let $ L \in \QMGFPyuL{m+1}{c}{s} $,
let \HiQndmIP{V} be the verifier and
\HiQndmIP{P} be the (honest or dishonest) prover.
We now describe what happens in each but the last round,
so we can give a name to all quantum registers in the
process.
\par
In the beginning of round $i$ for $i \in \BXycapFx{m}$ the prover is holding
some private register from the previous round which
we call $\nGRKiFQI{P}_{i-1}$.
The verifier is holding the answer from the previous
round ($\nGRKiFQI{A}_{i-1}$) and his private register
($\nGRKiFQI{V}_{i-1}$).
Note that $\nGRKiFQI{A}_{i-1}$ is made up of one qubit
and $\nGRKiFQI{V}_{i-1}$ is \nkZOryEx{n} qubits.
The verifier applies a unitary transformation
$\FdVGMifg{W}_i$ on $\nGRKiFQI{A}_{i-1} \nGRKiFQI{V}_{i-1}$,
and gets the registers $\nGRKiFQI{Q}_i$ and
$\nGRKiFQI{V}_i$ as the output.
The one qubit $\nGRKiFQI{Q}_i$ holds the $i$th question,
while the \vpmfnZuX{}-qubit $\nGRKiFQI{V}_i$ is the verifier's
new private register.
The verifier then sends $\nGRKiFQI{Q}_i$ to the prover, who
applies the unitary $\FdVGMifg{U}_i$ on
$\nGRKiFQI{P}_{i-1} \nGRKiFQI{Q}_i$ and gets
$\nGRKiFQI{P}_i$ and $\nGRKiFQI{A}_i$ as the output.
At the end of the round the prover sends back $\nGRKiFQI{A}_i$
to the verifier.
Figure~\ref{fig:round_description_QIP} shows
a schematic of this procedure.
Note that it is \EtVQAUvQ{} that for all $i,j \in \BXycapFx{m}$,
$\nGRKiFQI{P}_i$ has the same number of quibts as
$\nGRKiFQI{P}_j$, and the same is true for all $\nGRKiFQI{V}_i$s.
Moreover all $\nGRKiFQI{Q}_i$s and $\nGRKiFQI{A}_i$s are
one qubits.
Since we introduce new registers whenever an operation takes place,
we can talk about `the state of a register' without confusion.
\UdxPemrT{} we set the sate of $\nGRKiFQI{P}_0 \nGRKiFQI{A}_0 \nGRKiFQI{V}_0$
to be \CligxKFo{00 \ldots 0}.
We didn't put any upper bound on the size of register $\nGRKiFQI{P}_i$.
However, the prover doesn't need arbitrary big private space,
neither in the honest nor in the dishonest case.
This is formalized by the following lemma.
\begin{figure}[h]
\centering
\begin{asy}
size(14cm);
defaultpen(fontsize(11));
void drawbox(int x, int y, int height, string name, string index){
draw((x, y)--(x+2, y)--(x+2, y+height)--(x, y+height)--cycle);
label(
"$\mathsf{" + name + "}_{" + index + "}$",
(x+1,y+(height/2)));
}
void drawarrow(real x1, int x2, int y, string name){
draw((x1, y)--(x2 - 0.1, y), EndArrow(SimpleHead));
label(name, ((x1+x2-0.1)/2, y), N);
}
int stepd = 6;
draw((-2, 7)--(4*stepd+4, 7), dashed);
drawbox(0, 0, 4, "V", "i-1");
drawbox(0, 4, 2, "A", "i-1");
drawbox(0, 10, 6, "P", "i-1");
drawbox(stepd, 0, 4, "V", "i");
drawbox(stepd, 4, 2, "Q", "i");
drawbox(stepd, 10, 6, "P", "i-1");
drawbox(stepd*2, 0, 4, "V", "i");
drawbox(stepd*2, 8, 2, "Q", "i");
drawbox(stepd*2, 10, 6, "P", "i-1");
drawbox(stepd*3, 0, 4, "V", "i");
drawbox(stepd*3, 8, 2, "A", "i");
drawbox(stepd*3, 10, 6, "P", "i");
drawbox(stepd*4, 0, 4, "V", "i");
drawbox(stepd*4, 4, 2, "A", "i");
drawbox(stepd*4, 10, 6, "P", "i");
draw((stepd+2.1,5){right}..{right}(stepd*2-0.1,9), EndArrow(SimpleHead));
draw((stepd*3+2.1,9){right}..{right}(stepd*4-0.1,5), EndArrow(SimpleHead));
drawarrow(3.4, stepd, 3, "$ \mathbf{W}_i $");
drawarrow(2*stepd + 3.4, 3*stepd, 12, "$ \mathbf{U}_i $");
curly_bracket(2.2, 0, 6);
curly_bracket(2*stepd + 2.2, 8, 16);
label("$V$:", (-1, 3), W);
label("$P$:", (-1, 12), W);
\end{asy}
\caption{Schematic of round $i$ in a $ \WfUTAOJK_{\mathrm{short}} $
proof system.}
\label{fig:round_description_QIP}
\end{figure}
\begin{lem}
Fix any verifier strategy \HiQndmIP{V} and prover strategy
\HiQndmIP{P} for an input to the problem $L$, as described above.
Then we can construct another prover $\HiQndmIP{P}'$ which
makes the verifier accept with exactly the same probability
as \HiQndmIP{P},
and if \HiQndmIP{V} interacts with $\HiQndmIP{P}'$
then for all $i \in \BXycapFx{m}$, at most $2i$ qubits of
$\nGRKiFQI{P}_i$ will have state different than \CligxKFo{0}.
\label{lem:prover_space_compression}
\end{lem}
\begin{proof}
We prove the lemma by induction on $i$.
\UgLlztAi{} we modify \HiQndmIP{P} round-by-round
to satisfy the statement of the lemma.
If $i=0$ then $\nGRKiFQI{P}_0$ already has all
qubits in state \CligxKFo{0}, so we are done.
\par
Now suppose that for some $i$, we have that at most
$2i$ qubits of $\nGRKiFQI{P}_i$ has state different than \CligxKFo{0}.
Then we will show that at most $2i+2$ qubits of $\nGRKiFQI{P}_{i+1}$
has state different than \CligxKFo{0}, by modifying \HiQndmIP{P},
\keQZCenD{} the acceptance probability will stay the same.
So let us consider the situation in round $\HAPnvheP{i+1}$,
right after the verifier sent $\nGRKiFQI{Q}_{i+1}$
to the prover.
See Figure~\ref{fig:space_compression_QIP}.
Let's denote the part of $\nGRKiFQI{P}_i$ that contains
something other then \CligxKFo{0} by $\nGRKiFQI{P}_i'$.
By the induction hypothesis, the size of $\nGRKiFQI{P}_i'$
is $\leq 2i$.
Since the joint state of
$\nGRKiFQI{P}_i' \nGRKiFQI{Q}_{i+1} \nGRKiFQI{V}_{i+1}$
is pure, there exist a unitary $\FdVGMifg{X}_{i+1}$
acting on $\nGRKiFQI{V}_{i+1}$ and transforming it
to $\nGRKiFQI{W}_{i+1}$ of size $\leq 2i+1$ and a register
containing only \CligxKFo{00 \ldots 0}.
Suppose \HiQndmIP{V} performs this operation and let us not
worry now about whether it is doable in polynomial time.
Now $\HiQndmIP{P}'$ performs $\FdVGMifg{U}_{i+1}$ and gets
$\nGRKiFQI{P}_{i+1}$ and $\nGRKiFQI{A}_{i+1}$.
Now the joint state of
$\nGRKiFQI{P}_{i+1} \nGRKiFQI{A}_{i+1} \nGRKiFQI{W}_{i+1}$
is pure, so by the same argument as before,
we have that there exist a unitary $\FdVGMifg{Y}_{i+1}$
that transforms $\nGRKiFQI{P}_{i+1}$ into
$\nGRKiFQI{P}_{i+1}'$ of size $\leq 2i+2$ and some
register containing only \CligxKFo{00 \ldots 0}.
Our $\HiQndmIP{P}'$ performs this $\FdVGMifg{Y}_{i+1}$ as well.
Now \HiQndmIP{V} performs \VGXhOmeS{\FdVGMifg{X}_{i+1}}
and gets back $\nGRKiFQI{V}_{i+1}$.
Additionally $\HiQndmIP{P}'$ in the next round will
perform \VGXhOmeS{\FdVGMifg{Y}_{i+1}} just before
it is about to perform $\FdVGMifg{U}_{i+2}$,
so at that point the state of the whole system
will be the same as in the original proof system, so
the acceptance probability won't change.
Note that \HiQndmIP{V} doesn't actually need to
perform $\FdVGMifg{X}_{i+1}$ since it is followed
by \VGXhOmeS{\FdVGMifg{X}_{i+1}}.
So we don't need to modify \HiQndmIP{V} at all.
This finishes the proof of the lemma.
\end{proof}
\begin{figure}[h]
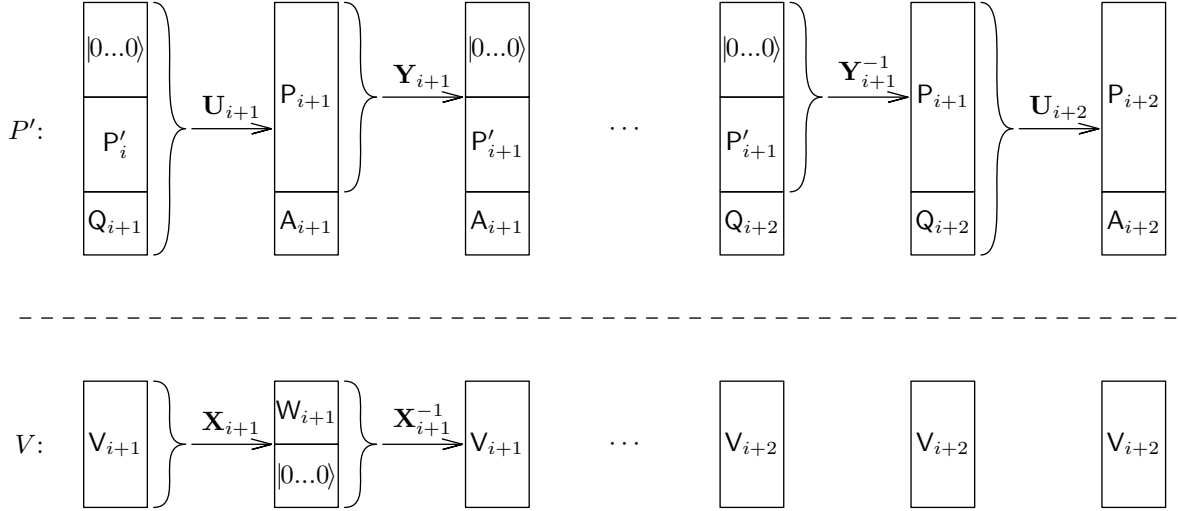

\centering
\begin{asy}
size(15.7cm);
defaultpen(fontsize(10));
int stepd = 6;
void drawbox(real x, int y, int height, string name){
draw((x, y)--(x+2, y)--(x+2, y+height)--(x, y+height)--cycle);
label(name, (x+1, y+(height/2)));
}
void drawbox_reg(real x, int y, int height, string name, string index){
drawbox(x, y, height, "$\mathsf{" + name + "}_{" + index + "}$");
}
void drawbox_regp(real x, int y, int height, string name, string index){
drawbox(x, y, height, "$\mathsf{" + name + "}'_{" + index + "}$");
}
void drawbox_zero(real x, int y, int height){
drawbox(x, y, height,
"$ \left| \: \! \! 0 ... 0 \: \! \! \right\rangle $");
}
void bracketarrow(real x, int ybottom, int ytop, string op, string subscript, bool inverse = false){
real ym = (ybottom + ytop) / 2;
curly_bracket(x + 0.2, ybottom, ytop);
real x_begin = x + 1.4;
real x_end = x + stepd - 2.1;
draw((x_begin, ym)--(x_end, ym), EndArrow(SimpleHead));
string inv = "";
if(inverse) inv = "^{-1}";
label(
"$ \mathbf{" + op + "}" + inv + "_{" + subscript + "} $",
((x_begin + x_end) / 2, ym), N);
}
real bigstep = (stepd - 2) * 1.5 + 2;
draw((-2, 6)--(4 * stepd + bigstep + 3, 6), dashed);
drawbox_reg(0, 0, 4, "V", "i+1");
drawbox_reg(0, 8, 2, "Q", "i+1");
drawbox_regp(0, 10, 3, "P", "i");
drawbox_zero(0, 13, 3);
bracketarrow(2, 0, 4, "X", "i+1");
bracketarrow(2, 8, 16, "U", "i+1");
drawbox_zero(stepd, 0, 2);
drawbox_reg(stepd, 2, 2, "W", "i+1");
drawbox_reg(stepd, 8, 2, "A", "i+1");
drawbox_reg(stepd, 10, 6, "P", "i+1");
bracketarrow(stepd+2, 0, 4, "X", "i+1", true);
bracketarrow(stepd+2, 10, 16, "Y", "i+1");
drawbox_reg(2*stepd, 0, 4, "V", "i+1");
drawbox_reg(2*stepd, 8, 2, "A", "i+1");
drawbox_regp(2*stepd, 10, 3, "P", "i+1");
drawbox_zero(2*stepd, 13, 3);
label("$ \cdots $", (2 * stepd + bigstep / 2 + 1, 2));
label("$ \cdots $", (2 * stepd + bigstep / 2 + 1, 12));
drawbox_reg(2 * stepd + bigstep, 0, 4, "V", "i+2");
drawbox_reg(2 * stepd + bigstep, 8, 2, "Q", "i+2");
drawbox_regp(2 * stepd + bigstep, 10, 3, "P", "i+1");
drawbox_zero(2 * stepd + bigstep, 13, 3);
bracketarrow(2 * stepd + bigstep + 2, 10, 16, "Y", "i+1", true);
drawbox_reg(3 * stepd + bigstep, 0, 4, "V", "i+2");
drawbox_reg(3 * stepd + bigstep, 8, 2, "Q", "i+2");
drawbox_reg(3 * stepd + bigstep, 10, 6, "P", "i+1");
bracketarrow(3 * stepd + bigstep + 2, 8, 16, "U", "i+2");
drawbox_reg(4 * stepd + bigstep, 0, 4, "V", "i+2");
drawbox_reg(4 * stepd + bigstep, 8, 2, "A", "i+2");
drawbox_reg(4 * stepd + bigstep, 10, 6, "P", "i+2");
label("$V$:", (-1, 2), W);
label("$P'$:", (-1, 12), W);
\end{asy}
\caption{Modification of the prover in round
$ i+1 $, in the proof of Lemma~\ref{lem:prover_space_compression}.}
\label{fig:space_compression_QIP}
\end{figure}
\begin{cor}
\UdxPemrT{}, for all $i \in \BXycapFx{m}$ we can assume that
$\nGRKiFQI{P}_i$ is made up of at most $2m$ qubits,
both in the honest and dishonest case.
Furthermore, the action of the prover in round $i$
is still a unitary, transforming
$ \nGRKiFQI{P}_{i-1} \nGRKiFQI{Q}_i $ to
$ \nGRKiFQI{P}_i \nGRKiFQI{A}_i $.
\label{cor:upper_bound_on_provers_space}
\end{cor}
\section{Proof of the main theorem}
\label{sec:main_theorem_proof}
This section finishes the proof of the main
theorem by using the result from the previous
section.
The idea is that the prover's unitaries in the
first $m$ rounds can be given as classical
descriptions of quantum circuits.
Using this, the \VFQyzTUk{} verifier can approximately
produce the state of the whole system appearing
before the last round in the \WfUTAOJK{} protocol.
This means the prover's private space, the answer
to the verifier and the verifier's private space.
To simulate the last round, we don't need to care
about the prover's private space, so we treat
its operation as a quantum channel, acting
on the the private space of the prover and the
question from the verifier.
Since the input is on \EtgXbDyQ-many qubits,
to perform the action of this channel, we can use
the same method as in \cite[Section~3]{Beigi2011}.
The detailed proof is as follows.
\begin{proof}[Proof of Theorem~\ref{thm:QIPshort_equals_QMA}]
The inclusion $ \VFQyzTUk \subseteq \QMGFPyuL{\OraXXrdX{n}}{c}{s} $
is trivial, so we only need to prove
$ \QMGFPyuL{\OraXXrdX{n}}{c}{s} \subseteq \VFQyzTUk $.
Just as above, let $ L \in \QMGFPyuL{m+1}{c}{s} $,
where $ m = \WctPVadE{n} $, and let \HiQndmIP{V} be the
corresponding verifier.
We will construct a verifier \HiQndmIP{W} for
the \VFQyzTUk{} proof system.
Because of Corollary~\ref{cor:upper_bound_on_provers_space},
we can assume that any prover strategy in the first
$m$ rounds are unitary operators on $2m$ qubits, say
$ \FdVGMifg{U}_1, \ldots, \FdVGMifg{U}_m $.
The constructed \HiQndmIP{W} expects to get as part
of the proof, the classical descriptions of circuits
$ \HiQndmIP{C}_{\FdVGMifg{U}_1, 1 / 3^n}, \ldots,
\HiQndmIP{C}_{\FdVGMifg{U}_m, 1 / 3^n} $,
\DwcWexfp{} the circuits that approximate the prover's
operators with precision $1/3^n$.
According to Corollary~\ref{cor:unitary_approx_by_circuit}
the length of this proof is
$ \UpcxRgfA{m \cdot 5^{2m} \cdot \GPnJfPwC{\log^3}{5^{2m} \cdot 3^n} }
\in \nkZOryEx{n} $.
\HiQndmIP{W} uses this classical proof to simulate the
first $m$ rounds of the proof system,
and produce the state of registers
$\nGRKiFQI{P}_m \nGRKiFQI{A}_m \nGRKiFQI{V}_m$, which
we denote by \CligxKFo{\psi}.
Note that since each circuit approximates
the corresponding unitary with precision $1/3^n$,
after applying \EtgXbDyQ{}-many of them,
it is true that
\[ \fsXxynbX{ \sCUPypWN{\psi} }{ \sCUPypWN{\phi} } \leq
\frac{m}{3^n} \leq \frac{1}{2^n} , \]
for sufficiently large $n$;
where \CligxKFo{\phi} is the state of
$\nGRKiFQI{P}_m \nGRKiFQI{A}_m \nGRKiFQI{V}_m$
in the case where the unitaries
$ \FdVGMifg{U}_1, \ldots, \FdVGMifg{U}_m $
were applied instead of the circuits.
\par
We are left with specifying how \HiQndmIP{W}
simulates the prover in the last $\HAPnvheP{m+1}$th
round.
We use exactly the same method as in the proof of
Theorem~\ref{thm:QIPlogpoly_equals_QMA},
appeared in \cite{Beigi2011}, and our proof
closely follows that proof as well.
Since we are in the last round,
we don't have to keep track of the prover's
private space, so we can just describe it's strategy
as a quantum channel that transforms registers
$\nGRKiFQI{P}_m \nGRKiFQI{Q}_{m+1}$ to $\nGRKiFQI{A}_{m+1}$.
Let's call this channel
$\Phi \in \zRYfDEHl{\zkKUsJUE{S}}{\zkKUsJUE{R}}$ from now on;
where \zkKUsJUE{S} is the joint space associated to
registers $\nGRKiFQI{P}_m \nGRKiFQI{Q}_{m+1}$,
and \zkKUsJUE{R} is the space associated to $\nGRKiFQI{A}_{m+1}$.
The input space \zkKUsJUE{S} is on
$ q \XYYXlqUI 2m+1 = \UpcxRgfA{\log n} $
qubits and the output space \zkKUsJUE{R} is on \nkZOryEx{n} qubits.
\HiQndmIP{W} expects to get
$ \rho_{\Phi}^{\RoHOARwu \HAPnvheP{N+k}} $
as the quantum part of it's proof,
where $ \rho_{\Phi} \in \ppKvmYcj{\zkKUsJUE{R} \RoHOARwu \zkKUsJUE{S}} $
is the normalized \LsetZeCP{} of $\Phi$,
for $N$ and $k$ to be specified later.
Let's divide up the quantum certificate given to
\HiQndmIP{W} into registers
$\nGRKiFQI{R}_1$, $\nGRKiFQI{S}_1$,
$\nGRKiFQI{R}_2$, $\nGRKiFQI{S}_2$, \ldots,
$\nGRKiFQI{R}_{N+k}$, $\nGRKiFQI{S}_{N+k}$,
where the space of each $\nGRKiFQI{R}_i$ is \zkKUsJUE{R},
and the space of each $\nGRKiFQI{S}_i$ is \zkKUsJUE{S}.
\HiQndmIP{W} expects each $\nGRKiFQI{R}_i \nGRKiFQI{S}_i$
to contain a copy of $\rho_{\Phi}$.
To simulate the last round of the interactive proof system,
\HiQndmIP{W} does the following.
\begin{enumerate}
\item \label{step:permute_and_discard}
Randomly permute the pairs
$ \HAPnvheP{\nGRKiFQI{R}_1, \nGRKiFQI{S}_1}, \ldots,
\HAPnvheP{\nGRKiFQI{R}_{N+k}, \nGRKiFQI{S}_{N+k}} $, according
to a uniformly chosen permutation,
and discard all but the first
\HAPnvheP{N+1} pairs.
\item \label{step:tomography}
Perform quantum state tomography on the
registers \HAPnvheP{\nGRKiFQI{S}_2, \ldots, \nGRKiFQI{S}_{N+1}},
and reject if the resulting approximation is not within
trace-distance $\delta / 2$ of the completely mixed state
$\HAPnvheP{1/2^q} \fAEoNWeb$, for $\delta$ to be specified below.
\item \label{step:simulation}
Simulate the channel specified by
\HAPnvheP{\nGRKiFQI{R}_1, \nGRKiFQI{S}_1} by
post-selection.
Reject if post-selection fails,
otherwise simulate the last operation of \HiQndmIP{V}
and accept \bsIbxBHW{} \HiQndmIP{V} accepts.
\end{enumerate}
Let $ \GPnJfPwC{g}{n} \in \nkZOryEx{n} $ be \keQZCenD{}
$\GPnJfPwC{c}{n} - \GPnJfPwC{s}{n} \geq 1 / \GPnJfPwC{g}{n}$.
We now set the parameters.
\[ \varepsilon \XYYXlqUI \frac{1}{4^{q+1}g}, \quad
\delta \XYYXlqUI \frac{\varepsilon^2}{4}, \quad
N \XYYXlqUI \orxHiCqI{\frac{2^{10q}}{\HAPnvheP{\delta/2}^3}}, \quad
k \XYYXlqUI \orxHiCqI{ \frac{4^{q+1} \HAPnvheP{N+1}}{\varepsilon} }
\]
Note that $ 1/\varepsilon, 1/\delta, N, k \in \nkZOryEx{n} $.
\paragraph{Completeness.}
Suppose there exist a \HiQndmIP{P} that causes
\HiQndmIP{V} to accept with probability $ \geq c $.
Let the certificate to \HiQndmIP{W} be the classical
descriptions of circuits
$ \HiQndmIP{C}_{\FdVGMifg{U}_1, 1 / 3^n}, \ldots,
\HiQndmIP{C}_{\FdVGMifg{U}_m, 1 / 3^n} $,
together with the state
$ \rho_{\Phi}^{\RoHOARwu \HAPnvheP{N+k}} $,
where each $\nGRKiFQI{R}_i \nGRKiFQI{S}_i$
contains a copy of $\rho_{\Phi}$,
for $ i \in \BXycapFx{N+k} $.
After simulating the first $m$ rounds,
\HiQndmIP{W} produces \CligxKFo{\psi} which is
$ \leq 1/2^n $ far from the correct \CligxKFo{\phi}
in the trace distance, just as described above.
Note that in the simulation of the last round,
step~\ref{step:permute_and_discard} doesn't
change the state of registers
$ \HAPnvheP{\nGRKiFQI{R}_1, \nGRKiFQI{S}_1}, \ldots,
\HAPnvheP{\nGRKiFQI{R}_{N+1}, \nGRKiFQI{S}_{N+1}} $.
According to Lemma~\ref{lem:state_tomography},
\HiQndmIP{W} rejects in step~\ref{step:tomography}
with probability $ \leq \delta/2 $.
In step~\ref{step:simulation}, post-selection
succeeds with probability $ 1/4^q $.
If \HiQndmIP{W} was using \CligxKFo{\phi} instead
of \CligxKFo{\psi} the probability of acceptance
would be at least
\[ \HAPnvheP{1 - \frac{\delta}{2}} \frac{c}{4^q} . \]
So using \CligxKFo{\psi}, the probability that \HiQndmIP{W}
accepts is at least
\[ \HAPnvheP{1 - \frac{\delta}{2}} \frac{c}{4^q} - \frac{1}{2^n}
\geq \frac{c}{4^q} - \varepsilon - \frac{1}{2^n} . \]
\paragraph{Soundness.}
Suppose that all \HiQndmIP{P} causes \HiQndmIP{V}
to accept with probability $ \leq s $.
Note that, \EtVQAUvQ{} any classical proof specifies some
set of unitaries that correspond to a valid prover strategy.
Hence it is still true, that after \HiQndmIP{W} simulates
the first $m$ rounds using the given circuits,
it ends up with a state \CligxKFo{\psi} that is
at most $ 1/2^n $ far from a state \CligxKFo{\phi},
where \CligxKFo{\phi} can be produced by some
\HiQndmIP{P} interacting with \HiQndmIP{V}.
\par
Now consider the situation that the state of
\HAPnvheP{\nGRKiFQI{S}_1, \ldots, \nGRKiFQI{S}_{N+1}}
before step~\ref{step:simulation} has the form
\begin{align}
\sigma^{\RoHOARwu \HAPnvheP{N+1}} ,
\label{eq:separable_state}
\end{align}
for some $ \sigma \in \ppKvmYcj{\zkKUsJUE{S}} $.
(The classical part of the proof has
been used up and discarded before
step~\ref{step:permute_and_discard}.)
We consider two cases:
\begin{itemize}
\item Suppose that
$ \RtRjhAaz{\sigma - \HAPnvheP{1/2^q} \fAEoNWeb} < \delta $.
Let the state of \HAPnvheP{\nGRKiFQI{R}_1, \nGRKiFQI{S}_1}
before step~\ref{step:simulation}
be $ \xi \in \ppKvmYcj{\zkKUsJUE{R} \RoHOARwu \zkKUsJUE{S}} $,
so we have $ \jYPzSEcs{\zkKUsJUE{R}}{\xi} = \sigma $.
Because of the same argument as in
\cite{Beigi2011}, there exists a state
$ \tau \in \ppKvmYcj{\zkKUsJUE{R} \RoHOARwu \zkKUsJUE{S}} $
\keQZCenD{} $ \jYPzSEcs{\zkKUsJUE{R}}{\tau} = \HAPnvheP{1/2^q} \fAEoNWeb $
and $ \fsXxynbX{\tau}{\xi} \leq \varepsilon $.
Given this $\tau$, the post-selection in
step~\ref{step:simulation} succeeds
with probability $ 1/4^q $, so the acceptance
in step~\ref{step:simulation} occurs with probability
at most $ s/4^q + 1/2^n $.
Given $\xi$ instead of $\tau$, \HiQndmIP{W}
will accept with probability at most
\[ \frac{s}{4^q} + \frac{1}{2^n} + \varepsilon . \]
\item If $ \RtRjhAaz{\sigma - \HAPnvheP{1/2^q} \fAEoNWeb} \geq \delta $,
then in step~\ref{step:tomography}, \HiQndmIP{W}
will accept with probability $ \leq \delta / 2 $.
\end{itemize}
Since $ \delta / 2 \leq s/4^q + 1/2^n + \varepsilon $ then in
both cases acceptance occurs with probability
$ \leq s/4^q + 1/2^n + \varepsilon $.
\par
Now suppose that the state of
\HAPnvheP{\nGRKiFQI{S}_1, \ldots, \nGRKiFQI{S}_{N+1}}
before step~\ref{step:simulation}
has the form
\begin{align}
\sum_i p_i \sigma_i^{\RoHOARwu \HAPnvheP{N+1}} ,
\label{eq:convex_combination}
\end{align}
for some probability vector $p$ and some set
$ \ZcdnmWHx{\sigma_i} \subset \ppKvmYcj{\zkKUsJUE{S}} $.
Since \eqref{eq:convex_combination} is a convex
combination of states of the form \eqref{eq:separable_state},
acceptance will occur with probability
$ \leq s/4^q + 1/2^n + \varepsilon $.
In the real scenario, by Theorem~\ref{thm:deFinetti},
it is true that the state of
\HAPnvheP{\nGRKiFQI{S}_1, \ldots, \nGRKiFQI{S}_{N+1}}
after step~\ref{step:permute_and_discard} will be
$\varepsilon$ close to a state of the form
\eqref{eq:convex_combination}, in the trace distance.
So the probability of acceptance of \HiQndmIP{W}
will be $ \leq s/4^q + 2 \varepsilon + 1/2^n $.
Since \[ \frac{c}{4^q} - \varepsilon - \frac{1}{2^n}
- \HAPnvheP{ \frac{s}{4^q} + 2 \varepsilon + \frac{1}{2^n} }
\geq \frac{1}{\GPnJfPwC{h}{n}} , \]
for some $ \GPnJfPwC{h}{n} \in \nkZOryEx{n} $,
it holds that $ L \in \VFQyzTUk $.
\end{proof}
\MGEHGnhN
\end{document}